\documentclass[3p,10pt,a4paper,twoside,fleqn]{elsarticle}
\usepackage{amssymb,amsmath,latexsym}
\usepackage{pst-all}
\usepackage[varg]{pxfonts}
\usepackage{mathrsfs}
\newtheorem{theorem}{Theorem}[section]
\newtheorem{lemma}[theorem]{Lemma}
\newtheorem{corollary}[theorem]{Corollary}

\newtheorem{example}[theorem]{Example}

\def\endproof{\qed\endtrivlist}
\expandafter\let\csname endproof*\endcsname=\endproof

\def\qedsymbol{\ifmmode\bgroup\else$\bgroup\aftergroup$\fi
  \vcenter{\hrule\hbox{\vrule height.6em\kern.6em\vrule}\hrule}\egroup}
\def\qed{\ifmmode\else\unskip\nobreak\fi\quad\qedsymbol}

\renewcommand{\iff}{\Leftrightarrow}

\renewcommand{\le}{\leqslant}

\newcommand{\im}{\qopname\relax{no}{Im}}

\newcommand{\lBrack}{\lbrack\!\lbrack}
\newcommand{\rBrack}{\rbrack\!\rbrack}

\newcommand{\oobslash}{\circledbslash\kern-10.2pt\bigcirc}
\newcommand{\ooslash}{\oslash\kern-10.2pt\bigcirc}

\usepackage{color}

\begin{document}

\journal{Fuzzy Sets and Systems}

\title{\Large\bf Brzozowski type determinization for fuzzy automata\tnoteref{t1}}
\tnotetext[t1]{Research supported by Ministry of Education, Science and Technological Development, Republic of Serbia, Grant No. 174013}
\author{Zorana Jan\v ci\' c}
\ead{zoranajancic329@gmail.com}

\author{Miroslav \'Ciri\'c\corref{cor}}
\ead{miroslav.ciric@pmf.edu.rs}

\cortext[cor]{Corresponding author. Tel.: +38118224492; fax: +38118533014.}
\address{University of Ni\v s, Faculty of Sciences and Mathematics, Vi\v segradska 33, 18000 Ni\v s, Serbia}

\begin{abstract}
In this paper we adapt the well-known Brzozowski determinization method to fuzzy automata.~This~method gives better results than all previously known methods for determinization of fuzzy auto\-mata developed by B\v elohl\'avek [Inform Sciences 143 (2002) 205--209], Li and Pedrycz [Fuzzy Set Syst 156 (2005) 68--92], Ignjatovi\'c et al. [Inform Sciences 178 (2008) 164--180], and Jan\v ci\'c et al. [Inform Sciences 181 (2011) 1358--1368].~Namely, as in the case of ordinary nondeterministic automata, Brzozowski type determinization~of~a~fuzzy automaton results in a minimal crisp-deterministic fuzzy automaton equivalent to the starting fuzzy automaton, and we show that there are cases when all previous methods result in infinite automata, while Brzozowski type deter\-minization results in a finite one.~The paper deals with fuzzy automata over complete residuated~lattices, but identical results can also be obtained in a more general context, for fuzzy automata over lattice-ordered monoids, and even for weighted automata over commutative semirings.
\end{abstract}


\maketitle

\section{Introduction}

The well-known Brzozowski's double reversal algorithm, presented for the first time in \cite{Brz.62}, is a concise and elegant algorithm having two purposes.~When its input is a nondeterministic automaton, the algorithm alternates two reverse and determinization operations (more precisely, the accessible subset construction) and produces a minimal deterministic automaton equivalent to the starting automaton.~In other words, the algorithm performs both determinization~and~mini\-mi\-zation.~On the other hand, when the input is a deterministic automaton, the algorithm performs its~minimization applying just one reverse and determinization operation.~Despite its worst-case exponential time complexity, the algorithm has recently gained popularity due to its excellent performance in practice, where it frequently outperforms theoretically faster algorithms (cf. \cite{AMR.07,AMR.13,TV.05,W.95}).~For more information about Brzozowski's double reversal algorithm, and~about algorithms for determinization of nondeterministic automata in general, we refer to \cite{CKP.02,Sak.09,vGP.08p,vGP.08tr}.

The purpose of this paper is to adapt Brzozowski's double reversal algorithm to fuzzy automata.~We start from an arbitrary fuzzy automaton and we show that applying twice the construction of a reverse Nerode automaton we obtain an equivalent minimal crisp-deterministic fuzzy automaton.~We also~demonstrate~that this fuzzy version of Brzozowski's double reversal algorithm outperforms all previous methods for determinization of fuzzy auto\-mata developed by B\v elohl\'avek \cite{Bel.02}, Li and Pedrycz \cite{LP.05}, Ignjatovi\'c~et~al.~\cite{ICB.08}, and Jan\v ci\'c et al.~\cite{JIC.11}, in the sense that it not only produces a smaller automaton than all the above mentioned methods, but even when all these methods produce infinite automata, Brzozowski type determinization can produce a finite one.~Moreover, when the starting fuzzy automaton is crisp-deterministic and accessible, its minimization is performed applying just one construction of a reverse Nerode automaton

The paper is organized as follows.~In the preliminary section we recall basic notions and notation concerning fuzzy sets and relations, fuzzy automata and languages and crisp-deterministic fuzzy automata, we recall the concept of a Nerode automaton and introduce the concept of a reverse Nerode automaton.~The main results are presented in Section 3.~We first introduce the notion of a right language associated with a state of a fuzzy automaton and describe some basic properties of right languages.~After that, we construct the right language automaton of a fuzzy automaton $\cal A$, we prove that it is isomorphic to the derivative automaton of the fuzzy language recognized by $\cal A$ (Theorem 3.3), and consequently, if all right fuzzy languages
associated with states of an accessible crisp-deterministic fuzzy automaton $\cal A$ are pairwise different, we show that $\cal A$ is minimal.~Then we prove that the
reverse Nerode automaton of any accessible crisp-deterministic fuzzy automaton $\cal A$ is a minimal crisp-deterministic fuzzy automaton equivalent to the reverse automaton of $\cal A$ (Theorem 3.5), and further, we define the concept of a Brzozowski automaton of a fuzzy automaton $\cal A$ and prove that it is a minimal crisp-deterministic fuzzy automaton~equivalent to $\cal A$ (Theorem 3.6).~Finally, we give a simple example of a fuzzy automaton $\cal A$ for which all previously known determinization methods produce an infinite crisp-deterministic fuzzy automaton, while Brzozowski automaton of $\cal A$ is finite and has only three states.

The most popular structure of membership values that has recently been used in the theory of fuzzy~sets, especially in the theory of fuzzy automata, are complete residuated lattices.~For this reason, this paper also deals with fuzzy automata over complete residuated lattices.~However, identical results can also be obtained in a more general context, for fuzzy automata over lattice-ordered monoids, and even for weighted automata over commutative semirings.

\section{Preliminaries}

\subsection{Fuzzy sets and relations}

In this work we will use complete residuated lattices as
structures of membership values.~A {\it residuated lattice\/} is
an algebra ${\cal L}=(L,\land ,\lor , \otimes ,\to , 0, 1)$ such
that
\begin{itemize}
\parskip=-2pt
\item[{\rm (L1)}] $(L,\land ,\lor , 0, 1)$ is a lattice with the
least element $0$ and the greatest element~$1$, \item[{\rm (L2)}]
$(L,\otimes ,1)$ is a commutative monoid with the unit $1$,
\item[{\rm (L3)}] $\otimes $ and $\to $ form an {\it adjoint
pair\/}, i.e., they satisfy the {\it adjunction property\/}: for
all $x,y,z\in L$,
\begin{equation}\label{eq:adj}
x\otimes y \le z \ \iff \ x\le y\to z .
\end{equation}
\end{itemize}
If, additionally, $(L,\land ,\lor , 0, 1)$ is a complete lattice, then ${\cal L}$ is called a {\it
complete residuated lattice\/}.

The operations $\otimes $ (called {\it multiplication\/}) and $\to
$ (called {\it residuum\/}) are intended for modeling the
conjunction and implication of the corresponding logical calculus,
and supremum ($\bigvee $) and infimum ($\bigwedge $) are intended
for modeling of the existential and general quantifier,
respectively.~For basic properties of complete residuated lattices we refer
to \cite{Bel.02b,BV.05}.

The most studied and applied structures of truth values, defined
on the real unit interval $[0,1]$ with\break $x\land y =\min
(x,y)$ and $x\lor y =\max (x,y)$, are the {\it {\L}ukasiewicz
structure\/} ($x\otimes y = \max(x+y-1,0)$, $x\to y=
\min(1-x+y,1)$), the {\it Goguen} ({\it product\/}) {\it
structure\/} ($x\otimes y = x\cdot y$, $x\to y= 1$ if $x\le y$
and~$=y/x$ otherwise) and the {\it G\"odel structure\/} ($x\otimes
y = \min(x,y)$, $x\to y= 1$ if $x\le y$ and $=y$
otherwise).
Another important set of truth values is the set
$\{a_0,a_1,\ldots,a_n\}$, $0=a_0<\dots <a_n=1$, with $a_k\otimes
a_l=a_{\max(k+l-n,0)}$ and $a_k\to a_l=a_{\min(n-k+l,n)}$. A
special case of the latter algebras is the two-element Boolean
algebra of classical logic with the support $\{0,1\}$.~The only
adjoint pair on the two-element Boolean algebra consists of the
classical conjunction and implication operations.~This structure
of truth values we call the {\it Boolean structure\/}.

In the sequel $\cal L$ will be a complete residuated
lattice.~A {\it fuzzy subset\/} of a set $A$ {\it over\/} ${\cal
L}$, or~simply a {\it fuzzy subset\/} of $A$, is any function from
$A$ into $L$.~Ordinary crisp subsets~of~$A$ are considered as
fuzzy subsets of $A$ taking membership values in the set
$\{0,1\}\subseteq L$.~Let $f$ and $g$ be two fuzzy subsets of
$A$.~The {\it equality\/} of $f$ and $g$ is defined as the usual
equality of mappings, i.e., $f=g$ if and only if $f(x)=g(x)$, for
every $x\in A$. The {\it inclusion\/} $f\le g$ is also defined
pointwise: $f\le g$ if and only if $f(x)\le g(x)$, for every $x\in
A$. Endowed with this partial order the set $L^A$ of all fuzzy
subsets of $A$ forms a complete residuated lattice, in which the
meet (intersection) $\bigwedge_{i\in I}f_i$ and the join (union)
$\bigvee_{i\in I}f_i$ of an arbitrary family $\{f_i\}_{i\in I}$ of
fuzzy subsets of $A$ are functions from $A$ into $L$ defined by
\[
\left(\bigwedge_{i\in I}f_i\right)(x)=\bigwedge_{i\in I}f_i(x), \qquad \left(\bigvee_{i\in
I}f_i\right)(x)=\bigvee_{i\in I}f_i(x),
\]
for every $x\in A$, and $f\otimes g$ and $f\to g$ are defined by $f\otimes g(x) = f (x)\otimes g(x)$ and $f\to g(x) = f (x)\to g(x)$, for all $f,g\in L^A$ and $x\in A$.

A {\it fuzzy relation\/} between sets $A$ and $B$ (in this order) is any mapping from
$A\times B$ into $L$, i.e., any fuzzy subset of $A\times B$, and the equality, inclusion (ordering), joins and meets
of fuzzy relations are defined as for fuzzy sets.~The set of all fuzzy relations between $A$ and $B$ will be denoted by $L^{A\times B}$.~In particular, a fuzzy relation on a set $A$ is any function from
$A\times A$ into $L$, i.e., any fuzzy subset of $A\times A$.  The {\it reverse\/} of a fuzzy relation $\varphi\in L^{A\times B}$ is a fuzzy
relation $\overline\varphi\in L^{B\times A}$ defined by $\overline\varphi(b, a) = \varphi(a, b)$, for all $a\in A$ and $b\in B$. A crisp relation is a fuzzy
relation which takes values only in the set $\{0, 1\}$, and if $\varphi$ is a crisp relation of $A$ to $B$, then expressions
$"\varphi(a, b) = 1"$ and $"(a, b)\in \varphi"$ will have the same meaning.

For non-empty sets $A$, $B$ and $C$, and fuzzy relations $\varphi\in L^{A\times B}$ and $\psi\in L^{B\times C}$,  their {\it
composition\/} is a fuzzy relation $\varphi\circ \psi\in  L^{A\times C}$ defined by
\begin{equation}\label{eq:comp.rr}
(\varphi \circ \psi )(a,c)=\bigvee_{b\in B}\,\varphi(a,b)\otimes \psi(b,c),
\end{equation}
for all $a\in A$ and $c\in C$.~Moreover, for $f\in L^A$, $\varphi \in L^{A\times B}$ and $g\in L^B$, compositions $f\circ\varphi\in L^B$ and $\varphi\circ g\in L^A$ and the scalar product $f\circ g\in L$ are defined by
\begin{equation}\label{eq:comp.sr}
(f \circ \varphi)(b)=\bigvee_{a'\in A}\,f(a')\otimes \varphi(a',b),\quad
(\varphi \circ g)(a)=\bigvee_{b'\in B}\,\varphi(a,b')\otimes g(b'),\quad
f \circ g =\bigvee_{a\in A}\,f(a)\otimes g(a) ,
\end{equation}
for all $a\in A$ and $b\in B$.

It is easy to check that $(\varphi_1\circ \varphi_2)\circ \varphi_3=\varphi_1\circ (\varphi_2\circ \varphi_3)$, $(f_1\circ\varphi_1)\circ\varphi_2=f_1\circ (\varphi_1\circ\varphi_2)$, $(f_1\circ\varphi_1)\circ f_2= f_1\circ (\varphi_1\circ f_2)$ and $(\varphi_1\circ\varphi_2)\circ f_1 = \varphi_1\circ(\varphi_2\circ f_1)$, for all fuzzy relations $\varphi_1$, $\varphi_2$ and $\varphi_3$ and fuzzy sets $f_1$ and $f_2$ for which  these compositions are defined,
and consequently, all parentheses in these expressions can be omitted.
Moreover, the composition of fuzzy relations is isotone in both arguments.

\subsection{Fuzzy automata}

In the further text, let ${\cal L}=(L,\land ,\lor , \otimes ,\to , 0, 1)$ be a complete residuated lattice and $X$ a finite alphabet.

A {\it fuzzy automaton over\/} $\cal L$ and $X$ , or simply a {\it fuzzy automaton\/}, is a quadruple
${\cal A}=(A,\delta,\sigma,\tau)$, where~$A$ is a non-empty set, called  the
{\it set of states\/}, $\delta:A\times X\times A\to L$~is~a
fuzzy subset~of $A\times X\times A$, called the~{\it fuzzy transition relation\/}, and $\sigma: A\to L$ and $\tau: A\to L$ are fuzzy subsets of~$A$, called the {\it fuzzy set of initial states} and the {\it fuzzy set terminal states}, respectively.~We can
interpret $\delta(a,x,b)$ as the degree~to~which an~input letter $x\in X$~causes~a~transition from a state $a\in A$ into a
state $b\in A$, whereas we can interpret $\sigma(a)$ and $\tau(a)$ as the degrees to which $a$ is respectively an input state and a terminal state. For methodological reasons we sometimes allow the set of states $A$ to be infinite.~A~fuzzy auto\-maton whose set of states is finite is called a
{\it fuzzy finite automaton\/}.~

Let $X^*$ denote the free monoid over the alphabet $X$, and  let $\varepsilon\in X^*$ be the
empty word.~The function~$\delta$~can be extended~up to a function
$\delta^*:A\times X^*\times A\to L$ as follows: For $a,b\in A$ and the empty word $\varepsilon$ we set
\begin{equation}\label{eq:delta.e}
\delta^*(a,\varepsilon ,b)=\begin{cases}\ \ 1, & \text{if}\ a=b, \\ \ \ 0, & \mbox{otherwise,}
\end{cases}
\end{equation}
and for $a,b\in A$, $u\in X^*$ and $x\in X$ we set
\begin{equation}\label{eq:delta.x}
\delta^*(a,ux ,b)= \bigvee _{c\in A} \delta^*(a,u,c)\otimes \delta (c,x,b).
\end{equation}
For each $u\in X^*$ we define a fuzzy relation $\delta_u\in L^{A\times A}$ by $\delta_u (a,b) = \delta^* (a,u,b)$, for all $a,b\in A$.~It is easy to check that  $\delta_{uv}= \delta_u\circ \delta_v$, for all $u,v\in X^*$.

A {\it fuzzy language\/} in $X^*$ over ${\cal L}$, or
briefly a {\it fuzzy language\/}, is any fuzzy subset of the free monoid $X^*$.~A {\it fuzzy language recognized by a fuzzy automaton\/} ${\cal A}=(A,\delta , \sigma,\tau )$ is a fuzzy language in $\lBrack{\cal A}\rBrack\in L^{X^*}$ defined by
\begin{equation}\label{eq:recog}
\lBrack{\cal A}\rBrack(u) = \bigvee_{a,b\in A} \sigma(a)\otimes \delta^*(a,u,b)\otimes \tau(b) = \sigma \circ \delta_u\circ \tau,
\end{equation}
for any $u\in X^*$.~In other words, the~membership degree of the word
$u$ to $\lBrack{\cal A}\rBrack$, i.e., the degree of recognition or acceptance of the word $u$, is equal to the degree to which $u$ leads from some initial to some terminal state.~Fuzzy automata $\cal A$~and~$\cal B$ are called {\it language equivalent\/}, or shortly just {\it equivalent\/}, if they recognize the same fuzzy language, i.e., if $\lBrack{\cal A}\rBrack=\lBrack{\cal B}\rBrack$.

For more information on the recognizability of fuzzy languages we refer to \cite{Boz.06,Boz.08,Boz.10}, and for infor\-ma\-tion on fuzzy automata over complete residuated lattices we refer to \cite{CSIP.10,IC.10,ICB.08,ICB.09,ICBP.10,SCI.11,Qiu.01,Qiu.06,Wechler.78,WQ.10,XQL.09}.

\subsection{Crisp-deterministic fuzzy automata}\label{sec:cdfa}

A {\it crisp-deterministic fuzzy  automaton\/} (for short: {\it cdfa\/})
over $X$ and ${\cal L}$ is a quadruple ${\cal A}=(A,\delta,a_0,\tau )$, where $A$ is a non-empty {\it set of states\/}, $\delta :A\times X\to A$ is a~{\it transition function\/}, $a_0\in A$ is an {\it initial state\/} and $\tau :A\to L$ is a {\it fuzzy set of final states\/}.~Equivalently, a crisp-deterministic
fuzzy automaton can be considered as a fuzzy~auto\-maton ${\cal A}=(A,\delta,\sigma,\tau )$ whose fuzzy transition function $\delta $ and fuzzy set of initial states
$\sigma $ satisfy the following conditions:
for all $x\in X$ and $a\in A$ there exists
$a'\in A$ such that $\delta_x(a,a')=1$, and $\delta_x(a,b)=0$, for all $b\in A\setminus
\{a'\}$,  and~ $\sigma(a_0)=1$, and $\sigma(a)=0$ for every
$a\in A\setminus \{a_0\}$. If the set of states $A$ is finite, then
$\cal A$ is called a
{\it crisp-deterministic fuzzy finite automaton\/} (for short: {\it cdffa\/}).

For a crisp-deterministic fuzzy  automaton ${\cal A}=(A,\delta,a_0,\tau )$,
the transition~function $\delta $ can be extended to a function $\delta^*:A\times X^*\to A$ by putting
$\delta^*(a,\varepsilon)=a$,  and $\delta^*(a,ux)=\delta (\delta^*(a,u),x)$, for
all $a\in A$, $u\in X^*$ and $x\in X$. A~state $a\in A$ is called {\it accessible\/} if there exists $u\in X^*$ such that
$\delta^*(a_0,u)=a$.~If every state of $\cal A$ is accessible, then $\cal A$ is called an {\it
accessible crisp-deterministic fuzzy automaton\/}.~The {\it fuzzy language\/} recognized
by $\cal A$ is the fuzzy language $\lBrack
{\cal A}\rBrack \in L^{X^*}$ given by
\begin{equation}\label{eq:cd-beh}
\lBrack {\cal A}\rBrack(u)=\tau (\delta^*(a_0,u)) \ ,
\end{equation}
for every $u\in X^*$.~Obviously, the image of $\lBrack {\cal A}\rBrack$ is contained in the image
of $\tau $, which is finite if the set of states is finite.~A fuzzy language $f\in L^{X^*}$ is called {\it cdffa-recognizable\/} if
there exists a crisp-deterministic fuzzy finite~auto\-maton $\cal A$ over $X$ and ${\cal L}$ such that
$\lBrack {\cal A}\rBrack=f $.

Let ${\cal A}=(A,\delta,a_0,\tau )$ and ${\cal A}'=(A',\delta',a_0',\tau' )$ be crisp-deterministic fuzzy automata. A function $\phi :A\to A'$ is called a {\it homomorphism\/} of $\cal A$ into ${\cal A}'$ if $\phi (a_0)=a_0'$, $\phi(\delta(a,x))=\delta'(\phi(a),x)$ and $\tau(a)=\tau'(\phi(a))$, for all~$a\in A$ and $x\in X$.~A bijective~homomor\-phism is called an {\it isomorphism\/}.~By $|{\cal A}|$ we denote the cardinality of the set of states of a fuzzy automaton $\cal A$. A crisp-deterministic fuzzy automaton ${\cal A}$ is called a {\it minimal crisp-deterministic fuzzy automaton} of a fuzzy language $f \in L^{X^*}$ if~it~recognizes $f$ and $|{\cal A} |\leqslant |{\cal A'}|$, for any crisp-deterministic fuzzy automaton ${\cal A'}$ which recognizes $f$.~Note that minimal crisp-deterministic fuzzy automata and minimization procedures that result in such automata were studied in \cite{ICBP.10,LP.07}.

For a fuzzy language $f\in L^{X^*}$ and $u\in X^*$, we define a fuzzy language $u^{-1}f\in L^{X^*}$ by $(u^{-1}f)(v)=f(uv)$, for each $v\in X^*$.~We call $u^{-1}f $  the {\it left derivative\/}~of~$f $ with respect to $u$.~Let $A_f =\{u^{-1}f\mid u\in X^*\}$
denote the set of all left derivatives of $f $, and let $\delta_f :A_f \times X\to
A_f $ and $\tau_f :A_f\to L$  be functions defined by
\begin{equation}\label{eq:delta-tau.phi}
\delta_f(g ,x)=x^{-1}g \ \ \ \text{and}\ \ \ \tau_f(g )=g (\varepsilon),
\end{equation}
for all $g \in A_f $ and $x\in X$.~Then ${\cal A}_f = (A_f , \delta_f
,f ,\tau_f )$ is an accessible  crisp-deterministic fuzzy automaton,~and~it is called the {\it left
derivative automaton\/}, or just the {\it derivative automaton\/}, of the fuzzy language $f $ \cite{IC.10,ICBP.10}.~It was proved in
\cite{ICBP.10} that the derivative automaton ${\cal A}_f $ is a minimal  crisp-deterministic fuzzy~auto\-maton which recognizes $f$, and therefore, ${\cal A}_f$ is finite if and only if the fuzzy language $f $ is cdffa-recognizable.~An algorithm for~construction of the
derivative automaton of a fuzzy language, based on simultaneous construction of the derivative automata of
ordinary languages $f^{-1}(a)$, for all $a\in \im (f )$,  was also given in~\cite{ICBP.10}.

\subsection{Nerode and reverse Nerode automaton}

Let ${\cal A}=(A,\delta,\sigma,\tau)$ be a fuzzy automaton over $\cal L$ and $X$.~The {\it reverse fuzzy
automaton\/} of ${\cal A}$ is a fuzzy automaton $\overline{\cal
A}=(A,\bar{\delta},\bar{\sigma},\bar{\tau})$, where $\bar{\sigma}=\tau$, $\bar{\tau}=\sigma$, and $\bar{\delta}:
A\times X\times A \to L$ is defined by:
\[
\bar{\delta}(a,x,b)=\delta(b,x,a),
\]
for all $a,b\in A$ and $x\in X$.~Roughly speaking, the reverse automaton of ${\cal A}$ is obtained from ${\cal A}$ by exchanging fuzzy sets of initial and final states and ``reversing'' all the transitions.

Due to the fact that the multiplication $\otimes$ is commutative, we have that $\bar{\delta}_{u}(a,b)=\delta_{\bar{u}}(b,a)$, for all $a,b\in A$ and $u\in X^*$.~For a fuzzy language $f\in L^{X^*}$, the {\it reverse fuzzy language\/} of $f$ is a fuzzy language $\overline{f}\in L^{X^*}$ defined by
$\overline{f}(u)=f(\bar{u})$, for each $u\in X^*$.~As $\overline{(\bar{u})} = u$ for all $u\in X^*$, we have that $\overline{(\overline{f})}=f$,
for each fuzzy language $f$.

If ${\cal A}$ is a fuzzy automaton over  ${\cal L}$ and $X$, it is easy to see that the reverse fuzzy automaton $\overline{\cal A}$ recognizes the reverse fuzzy language $\overline{\lBrack{\cal A}\rBrack}$ of the fuzzy language ${\lBrack{\cal A}\rBrack}$ recognized by $\cal A$, i.e., $\lBrack\overline{{\cal A}}\rBrack=\overline{\lBrack{\cal A}\rBrack}$.

Let ${\cal A}=(A,\delta,\sigma,\tau)$ be a fuzzy automaton over
$X$ and ${\cal L}$. For each $u\in X^*$ we define fuzzy sets $\sigma_u,\tau_u\in
L^A$ as follows:
\[
\sigma_{u}(a)=\bigvee_{b\in A}\sigma(b)\otimes\delta^*(b,u,a),\quad \tau_{u}(a)=\bigvee_{b\in A}\delta^*(b,u,a)\otimes\tau(b),
\]
for each $a\in A$. Equivalently,
\[
\sigma_{u}=\sigma\circ\delta_{u},\qquad \tau_u=\delta_u\circ \tau .
\]
The \textit{Nerode automaton} of ${\cal
A}=(A,\delta,\sigma,\tau)$ is a crisp-deterministic automaton ${\cal
A}_{N}=(A_{N},\delta_{N},\sigma_{\varepsilon},\tau_{N})$ whose set of states is $A_{N}=\{\sigma_{u}\mid
u\in X^{*}\}$, and functions $\delta_{N}:A_{N}\times X\longrightarrow A_{N}$ and
 $\tau_{N}:A_{N}\to L$ are defined by
\begin{equation}\label{eq:Naut}
\delta_{N}(\sigma_{u},x)=\sigma_{ux},\qquad\qquad \tau_{N}(\sigma_{u})=\sigma_{u}\circ\tau ,
\end{equation}
for all $u\in X^*$ and $x\in X$.~The concept of the Nerode automaton of a fuzzy automaton was first intro\-duced
by Ignjatovi\' c et al.~in \cite{ICB.08,ICBP.10}, for fuzzy automata over a complete residuated lattice, but it was also pointed out that the same construction can be extended to fuzzy automata over a lattice-ordered monoid, weighted automata over a semiring, and even to weighted automata over a strong bimonoid (cf.~\cite{CDIV.10,JIC.11}).~In \cite{ICB.08} it was also shown that the Nerode automaton of a fuzzy automaton ${\cal A}$ over a complete residuated
lattice is a crisp-deterministic fuzzy automaton equivalent to ${\cal A}$, i.e., $\lBrack{\cal A}_{N}\rBrack=\lBrack{\cal A}\rBrack$.

By the {\it reverse Nerode automaton\/} of  $\cal
A$ we will mean the Nerode automaton of the reverse fuzzy automaton $\overline{\cal A}$ of ${\cal A}$.~For the sake of simplicity, we denote the reverse Nerode automaton of $\cal A$ by ${\cal A}_{ \overline N}$ (instead of $(\overline{\cal A})_{N}$). Let us note that ${\cal A}_{\overline N}=(A_{\overline N},\delta_{\overline N},\tau_\varepsilon,\tau_{\overline N})$, where $A_{\overline N}=\{\tau_u\mid u\in
X^*\}$, and the functions $\delta_{\overline N}:A_{\overline N}\times X\to A_{\overline N}$ and $\tau_{\overline N}:A_{\overline N}\to L$ are given by
\begin{equation}\label{eq:rev.Naut}
\delta_{\overline N}(\tau_u,x)=\tau_{xu}, \qquad \tau_{\overline N}(\tau_u)=\tau_u\circ \sigma,
\end{equation}
for all $u\in X^*$ and $x\in X$.

\section{The main results}

Let ${\cal A}=(A,\delta,\sigma,\tau)$ be a fuzzy automaton over
$X$ and ${\cal L}$.

For any state $a\in A$, the {\it right fuzzy language associated with}  $a$ is the fuzzy language $\tau_{a}\in L^{X^*}$ defined
by
\[
\tau_{a}(u)=\bigvee_{b\in A} \delta^{*}(a,u,b)\otimes\tau(b),
\]
for each $u\in X^*$. In other words, $\tau_a$ is the fuzzy language recognized by a fuzzy automaton ${\cal
A}'=(A,\delta,a,\tau)$ obtained from ${\cal A}$ by replacing $\sigma $ with
the single crisp initial state $a$. The {\it left fuzzy language associated with}
$a$ is the fuzzy language $\sigma_{a}\in L^{X^*}$ given by
\[
\sigma_{a}(u)=\bigvee_{b\in A} \sigma(b)\otimes\delta^{*}(b,u,a),
\]
for each $u\in X^*$, i.e., the fuzzy language recognized by a fuzzy automaton ${\cal
A}'=(A,\delta,\sigma,\{a\})$ obtained from ${\cal A}$ by replacing $\tau $ with
the single crisp terminal state $a$.

We can easily show that the following is true.

\begin{lemma}\label{pr:rlll}
The right fuzzy language associated with the state $a$ of a fuzzy automaton
${\cal A}$ is equal to the reverse of the left fuzzy
language associated with the state $a$ in the reverse fuzzy automaton $\overline{\cal A}$.
\end{lemma}

For a crisp-deterministic fuzzy automaton
${\cal A}=(A,\delta,a_{0},\tau)$, the right fuzzy language associated with~a~state
$a$ in $\cal A$ is given by
\begin{equation}\label{eq:rlcd}
\tau_{a}(u)=\tau(\delta^{*}(a,u)),\qquad
\end{equation}
for each $u\in X^*$, and in particular, $\tau_{a_{0}}=\lBrack {\cal A} \rBrack$,
i.e., the right fuzzy language associated with the initial state~$a_{0}$ is the fuzzy language
recognized by ${\cal A}$.~It can be also easily verified that the following
is true.
\begin{lemma}\label{pr:rlcd}
Let ${\cal A}=(A,\delta,a_{0},\tau)$ be a crisp-deterministic fuzzy automaton. Then
\begin{equation}\label{eq:rlld}
\tau_{\delta^{*}(a,u)} = u^{-1}\tau_{a},
\end{equation}
for all $a\in A$ and $u\in X^*$.
\end{lemma}

If ${\cal A}=(A,\delta,a_{0},\tau)$ is an  crisp-deterministic fuzzy automaton, we define another  crisp-deterministic~fuzzy automaton ${\cal A}_{r}=(A_{r},\delta_{r},\tau_{a_{0}},\tau_{r})$ as follows:
the set of states $A_{r}$ is the set of all right fuzzy~languages associated with states of  ${\cal A}$, and  $\delta_{r}:A_{r}\times X\to A_{r}$ and $\tau_{r}:A_{r}\to L$ are given by:
\[
\delta_{r}(\tau_{a},x)=\tau_{\delta(a,x)}, \qquad \tau_{r}(\tau_{a})=\tau_{a}(\varepsilon),
\]
for each $\tau_a\in A_r$. We have the following:

\begin{theorem}\label{th:rlada}
Let ${\cal A}=(A,\delta,a_{0},\tau)$ be an accessible crisp-deterministic fuzzy automaton.~Then  ${\cal A}_{r}$ is an accessible crisp-deterministic fuzzy automaton isomorphic to the derivative automaton ${\cal A}_f$ of the fuzzy language $f=\lBrack {\cal A} \rBrack$.
\end{theorem}
\begin{proof}
Define a mapping $\phi: {A}_f \to {A}_{r}$ by
\[
\phi(u^{-1}f) = \tau_{\delta^{*}(a_{0},u)},
\]
for each $u\in X^*$.~If $u, v \in X^*$ such that $u^{-1}f = v^{-1}f$,
then according to (\ref{eq:rlld}) we obtain that
$\tau_{\delta^{*}(a_{0},u)} = \tau_{\delta^{*}(a_{0},v)}$, and hence $\phi(u^{-1}f) = \phi(v^{-1}f)$.~Thus, $\phi$ is well-defined.~On the other hand, let
$u,v\in X^*$ such that $\phi(u^{-1}f)=\phi(v^{-1}f)$, i.e., $\tau_{\delta^{*}(a_{0},u)}
=\tau_{\delta^{*}(a_{0},v)}$. Then by (\ref{eq:rlld}) it follows that
\[
u^{-1}f=u^{-1}\tau_{a_{0}}=\tau_{\delta^{*}(a_{0},u)}
=\tau_{\delta^{*}(a_{0},v)}=v^{-1}\tau_{a_{0}}=v^{-1}f.
\]
Therefore, $\phi$ is injective.~Due to the fact that ${\cal A}$ is accessible, it is easy to show that $\phi$ is a surjective mapping.

In order to prove that $\phi$ is a homomorphism, consider  arbitrary $u\in
X^*$ and $x\in X$. Then
\[
\phi(\delta_{f}(u^{-1}f,x))=\phi((ux)^{-1}f)= \tau_{\delta^{*}(a_{0},ux)}=\tau_{\delta(\delta^{*}(a_{0},u),x)}=\delta_{r}(\tau_{\delta^{*}(a_{0},u)},x)=\delta_{r}(\phi(u^{-1}f),x).
\]
Moreover,  $\phi(\varepsilon^{-1}f) = \tau_{a_{0}}$ and $\tau_{f}(u^{-1}f) = \tau_{r}(\phi(u^{-1}f))$. Hence, $\phi$ is an isomorphism.
\end{proof}

The automaton ${\cal A}_r$ will be called the {\it right language automaton\/} of $\cal A$.

By the previous theorem we obtain the following consequence.

\begin{corollary}\label{th:rl}
Let ${\cal A}=(A,\delta,a_0,\tau )$ be an accessible crisp-deterministic fuzzy automaton.~If all right fuzzy languages associated with states of $\cal A$ are pairwise different, then ${\cal A}$ is minimal.
\end{corollary}

\begin{proof}
It is clear that the function $\phi :A\to A_r$ defined by $\phi (a)=\tau_a$ is a homomorphism of $\cal A$ onto ${\cal A}_r$. Therefore, if all right fuzzy languages~associated with states of $\cal A$ are pairwise different, then $\phi $ is an isomorphism of $\cal A$ onto ${\cal A}_r$, and according to Theorem \ref{th:rlada}, ${\cal A}$ is minimal.
\end{proof}

Let us note that if ${\cal A}=(A,\delta,a_0,\tau )$ is a crisp-deterministic fuzzy automaton, then its reverse Nerode~auto\-maton is ${\cal A}_{\overline N}=(A_{\overline N},\delta_{\overline
N},\tau_\varepsilon,\tau_{\overline N})$, where $A_{\overline N}$ and $\delta_{\overline N}:A_{\overline N}\times X\to A_{\overline N}$ have the same form as in the
general case, whereas the function $\tau_{\overline N}:A_{\overline N}\to L$ is given by
\begin{equation}\label{eq:taun-cd}
\tau_{\overline N}(\tau_u)=\tau_u(a_0),
\end{equation}
for each $u\in X^*$.

Now, we are ready to prove the following

\begin{theorem}\label{th:rdma}
For any accessible crisp-deterministic fuzzy automaton ${\cal A}=(A,\delta,a_{0},\tau)$, the reverse Nerode automaton ${\cal A}_{\overline N}$~is~a~minimal crisp-deterministic
fuzzy automaton equivalent to $\overline{\cal A}$.
\end{theorem}

\begin{proof}
As we have already noted, ${\cal A}_{\overline N}$ is a crisp-deterministic
fuzzy automaton equivalent to $\overline{\cal A}$, so it remains to show that it is minimal.~According to Corollary \ref{th:rl}, it is enough to prove that all right fuzzy
languages associated with states of ${\cal A}_{\overline N}$ are pairwise different.

Let $\tau_u,\tau_v\in A_{\overline N}$, where $u,v\in X^*$, be two different states of ${\cal A}_{\overline N}$.~Then there is
$a\in A$ such that $\tau_u(a)\ne \tau_v(a)$, and since $\cal A$ is accessible, there is $w\in X^*$ such that $a=\delta^*(a_0,w)$. According to (\ref{eq:rlcd})
we obtain that
\[
\tau_{\tau_u}(\overline w)=\tau_{\overline N}(\delta^*_{\overline N}(\tau_u,\overline
w))=\tau_{\overline N}(\tau_{wu})=\tau_{wu}(a_0)=\tau_{a_0}(wu)=\tau (\delta^*(a_0,wu))=\tau (\delta(\delta^*(a_0,w),u))=\tau(\delta^* (a,u)) = \tau_{u}(a),
\]
whence $\tau_{\tau_u}(\overline w)=\tau_{u}(a)$, and in the same way we show
that $\tau_{\tau_v}(\overline w)=\tau_{v}(a)$.~Since $\tau_u(a)\ne \tau_v(a)$,~we~conclude~that $\tau_{\tau_u}(\overline w)\ne \tau_{\tau_v}(\overline w)$, and hence,
$\tau_{\tau_u}$ and $\tau_{\tau_v}$ are different right fuzzy languages associated
with states of ${\cal A}_{\overline N}$.
\end{proof}

Let ${\cal A}$ be a fuzzy automaton over an alphabet $X$ and a complete residuated lattice  ${\cal L}$. The {\it Brzozowski automaton\/} of ${\cal A}$, in notation ${\cal A}_B$, is a fuzzy automaton obtained from $\cal A$ applying twice the construction of the reverse Nerode automaton, i.e.,
\[
{\cal A}_{B}=\bigl({\cal A}_{\overline N}\bigr)_{\overline N}=\bigl({\overline{({\overline{\cal A})_{N}}}}\bigr)_{N} .
\]

Now we are ready to state and prove the main result of this paper.

\begin{theorem}\label{th:bd}
Let ${\cal A}$ be a fuzzy automaton over an alphabet $X$ and a complete residuated lattice  ${\cal L}$. The Brzozowski automaton ${\cal A}_{B}$ is a minimal crisp-deterministic fuzzy automaton equivalent to ${\cal A}$.
\end{theorem}
\begin{proof}
This is an immediate consequence of Theorem 4.1 of \cite{ICB.08} and Theorem \ref{th:rdma}.

Namely,~according to Theorem 4.1 of \cite{ICB.08}, the reverse Nerode automaton ${\cal A}_{\overline N}=(\overline{\cal A})_{N}$ is a crisp-determi\-nistic fuzzy automaton equivalent to $\overline{\cal A}$, and its reverse Nerode automaton $\bigl({\cal A}_{\overline N}\bigr)_{\overline N}$ is a crisp-deterministic fuzzy automaton equivalent to ${\cal A}$. Therefore, by Theorem \ref{th:rdma}, ${\cal A}_B$ is a minimal crisp-deterministic fuzzy automaton equivalent to ${\cal A}$.
\end{proof}

Let ${\cal A}=(A,\delta,\sigma,\tau)$ be a fuzzy finite automaton with $n$ states and $m$ input letters, and suppose that the subsemiring ${\cal L}^*(\delta,\sigma,\tau )$ of the semiring $(L,\lor,\otimes,0,1)$, generated by all membership~values taken by $\delta $, $\sigma $ and $\tau $, is finite and has $k$ elements.~An algorithm which constructs the Nerode automaton of $\cal A$~was provided in \cite{ICB.08}, and it can easily be transformed to an algorithm which constructs the reverse Nerode~automaton~of~$\cal A$. Any of these algorithms builds the \emph{transition tree} of the crisp-deterministic fuzzy automaton that it constructs (Nerode or reverse Nerode), and it is an $m$-ary tree with at most $k^n$ internal vertices which correspond to the states of the automaton under construction.~Computationally most demanding part of the algorithm is the one where for any newly-constructed fuzzy set (a vertex of the tree) the algorithm checks whether it has~already been computed before.~The computational time of this part, and the whole algorithm, is~$O(mnk^{2n})$ (cf.~\cite{JMIC.14}). Therefore, the first round of the application of the double reversal procedure to $\cal A$ produces an automaton with at most $k^n$ states, and the computational time of this round is $O(mnk^{2n})$.~The second round may start from an exponentially larger automaton, but despite that, this round produces a minimal crisp-deterministic fuzzy automaton equivalent to $\cal A$, an automaton that is not greater than the Nerode automaton of $\cal A$, which can not have more than $k^n$ states.~Thus, the resulting transition tree can not have more than~$k^n$~internal~vertices, and consequently, the second round has the same computational time $O(mnk^{2n})$.~This means that the total computational time of the Brzozowski's double reversal algorithm for the fuzzy automaton $\cal A$ is $O(mnk^{2n})$, the same as for constructions of the Nerode and the reverse Nerode automaton of $\cal A$.

\smallskip

Finally, we give the following example.

\begin{example}\label{ex:PS1}\rm
Let ${\cal A}=(A,\delta,\sigma,\tau)$ be a fuzzy finite automaton over the alphabet $X=\{x\}$ and the Goguen (product) structure, given by the transition graph shown in Figure 1.

\begin{center}
\psset{unit=1.2cm}
\begin{pspicture}(-2.5,-0.8)(2.5,2.4)
\pnode(0,0){C}
\SpecialCoor
\rput([angle=0,nodesep=20mm,offset=0pt]C){\cnode{3mm}{A1}}
\rput([angle=90,nodesep=20mm,offset=0pt]C){\cnode{3mm}{A2}}
\rput([angle=180,nodesep=20mm,offset=0pt]C){\cnode{3mm}{A0}}
\rput([angle=180,nodesep=5mm,offset=0pt]A0){\pnode{I}}
\rput([angle=0,nodesep=5mm,offset=0pt]A1){\pnode{J}}
\rput([angle=0,nodesep=5mm,offset=0pt]A2){\pnode{K}}
\rput(A1){\small $a_1$}
\rput(A2){\small $a_2$}
\rput(A0){\small $a_0$}
\NormalCoor
\ncline{->}{I}{A0}\aput[0pt](.1){\scriptsize $1$}
\ncline{<-}{J}{A1}\bput[1pt](.1){\scriptsize $1$}
\ncline{->}{A0}{A2}\aput[0pt](.5){\scriptsize $x/1$}
\ncline{->}{A2}{A1}\aput[0pt](.5){\scriptsize $x/1$}
\ncline{->}{A0}{A1}\aput[0pt](.5){\scriptsize $x/0.5$}
\nccircle[angleA=0]{->}{A2}{0.3}\bput[-35pt](.10){\scriptsize $x/0.5$}
\nccircle[angleA=180]{->}{A1}{0.3}\bput[-32pt](.10){\scriptsize $x/1$}
\end{pspicture}\\
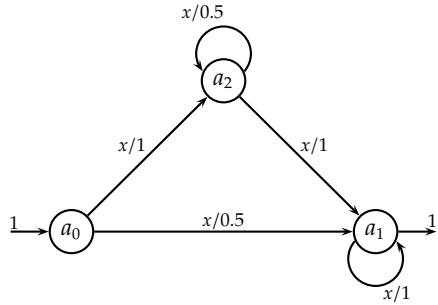
\small Figure 1. The transition graph of the fuzzy automaton ${\cal A}$
\end{center}
In matrix form, $\sigma$, $\delta_{x}$ and $\tau$ are represented as follows:
\[
\sigma=\begin{bmatrix}
1 & 0 & 0 \end{bmatrix},
\ \ \ \ \delta_x=\begin{bmatrix}
0 & 0.5 & 1\\
0 & 1 & 0\\
0 & 1 & 0.5 \end{bmatrix},\ \ \ \ \tau=\begin{bmatrix}
0 \\
1 \\
0
\end{bmatrix}.
\]
It is easy to verify that $\sigma_x=\begin{bmatrix}0 & 0.5 & 1\end{bmatrix}$ and
\[
\sigma_{x^n}=\begin{bmatrix} 0 & 1 & 0.5^{n-1}\end{bmatrix},
\]
for each $n\in \Bbb N$, $n\geqslant 2$, which means that the Nerode automaton of $\cal A$ has infinitely many states.

On the other hand, it is not hard to check that the reverse Nerode automaton ${\cal A}_{\overline N}$ and the Brzozowski automaton ${\cal A}_B$ are mutually isomorphic, and they are represented by the graph in Figure 2.
\begin{center}
\psset{unit=1.2cm}
\begin{pspicture}(-2.5,-0.8)(2.5,0.5)
\pnode(0,0){C}
\SpecialCoor
\rput([angle=0,nodesep=20mm,offset=0pt]C){\cnode{3mm}{A2}}
\rput([angle=0,nodesep=0mm,offset=0pt]C){\cnode{3mm}{A1}}
\rput([angle=180,nodesep=20mm,offset=0pt]C){\cnode{3mm}{A0}}
\rput([angle=180,nodesep=5mm,offset=0pt]A0){\pnode{I}}
\rput([angle=270,nodesep=5mm,offset=0pt]A0){\pnode{J}}
\rput([angle=270,nodesep=5mm,offset=0pt]A1){\pnode{L}}
\rput([angle=270,nodesep=5mm,offset=0pt]A2){\pnode{K}}
\rput(A2){\small $b_2$}
\rput(A1){\small $b_1$}
\rput(A0){\small $b_0$}
\NormalCoor
\ncline{->}{I}{A0}
\ncline{<-}{J}{A0}\bput[1pt](.1){\scriptsize $0$}
\ncline{<-}{L}{A1}\bput[2pt](.1){\scriptsize $0.5$}
\ncline{<-}{K}{A2}\bput[2pt](.1){\scriptsize $1$}
\ncline{->}{A0}{A1}\Aput[1pt]{\scriptsize $x$}
\ncline{->}{A1}{A2}\Aput[1pt]{\scriptsize $x$}
\nccircle[angleA=270]{<-}{A2}{0.4}\bput[0pt](.75){\qquad\scriptsize $x$}
\end{pspicture}\\
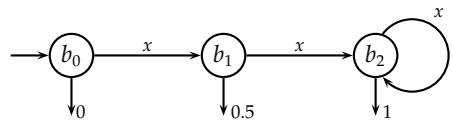
\small Figure 2. The transition graph of the reverse Nerode automaton ${\cal A}_{\overline N}$ and the Brzozowski automaton ${\cal A}_{B}$ of $\cal A$
\end{center}
\end{example}

This example demonstrates that there is a fuzzy finite automaton $\cal A$ whose Nerode automaton ${\cal A}_{N}$~is infinite, but Brzozowski automaton ${\cal A}_{\cal B}$ is finite.~Note that the determinization methods developed by~B\v elo\-hl\'avek \cite{Bel.02} and Li and Pedrycz \cite{LP.05} always result in automata whose cardinality is greater than or equal~to the cardinality of the related Nerode automaton (cf.~\cite{ICB.08}), and therefore, in this case these methods also~give infinite automata.~On the other hand, the method developed by Jan\v ci\'c et al.~\cite{JIC.11} always results in an~auto\-ma\-ton whose cardinality is less than or equal to the cardinality of the related Nerode automaton, but~according to Theorem 3.7 of \cite{JIC.11}, this automaton is finite if and only if the related Nerode automaton is finite.~Hence, in this case the method from \cite{JIC.11} also gives an infinite automaton.~Summing up, we conclude that all the above mentioned methods applied to the fuzzy finite automaton $\cal A$ from this example produce infinite automata, but Brzozowski automaton ${\cal A}_B$ is finite.

\section{Concluding remarks}

Brzozowski's double reversal algorithm is a well-known determinization-minimization algorithm which, despite its worst-case exponential time complexity, has excellent performance in practice and often~outper\-forms theoretically faster algorithms.~Here we have developed a Brzozowski type algorithm for fuzzy~autom\-ata.~We have shown that this algorithm outperforms all previously known methods for determinization of fuzzy automata, in the sense that it not only produces a smaller automaton than all previous methods, but even when all these methods produce infinite automata, Brzozowski type determinization can produce a finite one.~No matter that Brzozowski type algorithm has been developed here for fuzzy automata over complete residuated lattice, without any modifications it can be also applied to fuzzy automata over lattice-ordered monoids and weighted automata over commutative semirings.

In our future research, we will search for determinization methods that produce automata having even smaller number of states than the Nerode automaton or the reverse Nerode automaton.~Such an~improvement of the construction of the reverse Nerode automaton could also significantly improve the~performance of the double reversal algorithm for fuzzy automata.~In addition, it could be interesting to exploit the concept of an approximate equivalence of fuzzy automata, introduced in \cite{BK.09}, and study approximate~determini\-za\-tion of a fuzzy automaton, a procedure of constructing a crisp-deterministic fuzzy automaton whose~fuzzy~language is approximately equal to the fuzzy language of the given one.

\end{document}